\documentclass[12pt,a4paper,reqno]{amsart}
\usepackage{color}
\usepackage{amsmath}
\usepackage{amsfonts}
\usepackage{amssymb}
\usepackage{amsthm}
\usepackage{newlfont}
\usepackage{graphicx}

\newtheorem{thm}{Theorem}[section]
\newtheorem{cor}[thm]{Corollary}
\newtheorem{lem}[thm]{Lemma}
\newtheorem{prop}[thm]{Proposition}
\theoremstyle{definition}

\theoremstyle{remark}
\newtheorem{rem}[thm]{Remark}

\numberwithin{equation}{section}
\numberwithin{thm}{section}

\newcommand{\eps}{\varepsilon}

\newcommand{\lsm}{\lesssim}

\newcommand{\R}{{\mathbb{R}}}

\newcommand{\ed}{\end {document}}

\newcounter{smalllist}

\title{On a nonlocal aggregation model with nonlinear diffusion}
\author{Dong Li}
\address{School of Mathematics, Institute For Advanced Study, Princeton NJ 08544}%
\email{dongli@ias.edu}
\author{Xiaoyi Zhang}
\address{School of Mathematics, Institute For Advanced Study, Princeton NJ 08544}%
\email{xiaoyi@ias.edu}
\date{submitted Jun 2008}
\begin{document}
\maketitle
\begin{abstract}
We consider a nonlocal aggregation equation with nonlinear diffusion
which arises from the study of biological aggregation dynamics. As a
degenerate parabolic problem, we prove the well-posedness,
continuation criteria and smoothness of local solutions. For
compactly supported nonnegative smooth initial data we prove that the
gradient of the solution develops $L_x^\infty$-norm blowup in finite time.
\end{abstract}

\section{Introduction and main results}
In this paper we consider the following evolution equation with nonlinear diffusion:
\begin{align}\label{eq1}
\begin{cases}
 \partial_t u + \partial_x (u \partial_x K * u) = r\partial_x (u^2 \partial_x u), \quad (t,x)
  \in  (0,\infty)\times\R,\\
 u(0,x)=u_0(x),
\end{cases}
\end{align}
where $K$ is an even function of $x$ and has a Lipschitz point at the origin, e.g. $K(x) =e^{-|x|}$. 
Here $*$ denotes the usual spatial convolution on $\R$.
The parameter $r> 0$ measures
the strength of the nonlinear diffusion term. The function $u=u(t,x)$ represents population density in biology or particle density
in material science. 
(see, for example, \cite{TBL06} \cite{EkWG98} \cite{HP05} \cite{HP06} \cite{Mek99} \cite{ToBe04} \cite{LRC01} \cite{TT98}
\cite{FGLO99} \cite{Mu02} \cite{O80}). 
This equation was used in the study of biological aggregation such as insect swarms,
fish schools and bacterial colonies. It is first derived by Bertozzi, Lewis and Topaz \cite{TBL06} as a modification (more precisely, the addition of the
density-dependent term on the RHS of \eqref{eq1}) of an earlier classical model of Kawasaki \cite{K78}.
According to \cite{TBL06}, this modification gives rise to biologically meaningful clumping solutions
(i.e. densities with compact support and sharp edges). For other similar one-dimensional models and their biological applications we refer
the readers to \cite{Mek99} \cite{EkWG98} \cite{VCFH99} \cite{PD84} \cite{HM89} \cite{I84} \cite{I85} \cite{IN87} \cite{K78} \cite{MY82}
\cite{NM83} and the references therein for more extensive background and reviews.

To understand the biological meaning of each term, one can rewrite
\eqref{eq1} as the classical continuity equation
\begin{align}
 \partial_t u + \partial_x ( u v) =0,
\end{align}
where the velocity $v$ is related to the density $u$ by
\begin{align*}
 v &= \partial_x (K*  u) - r u\partial_x u \\
 &=: v_a +v_d.
\end{align*}
Here $v_a$ is called the attractive velocity since as explained in \cite{TBL06} individuals aggregate by climbing gradients of the sensing function
$s=K*u$. Due to the spatial convolution $v_a$ is a nonlocal transformation of the density $u$.
The second term $v_d$ is called the dispersal(repulsive) velocity and is a spatially local function of both the population and
the population gradient. Biologically $v_d$ represents the anti-crowding mechanism which operates in the opposite direction of population gradient
and decreases as the population density drops. These constitutive relations of $v_a$, $v_d$, $u$ are natural
 in view of the basic biological assumption that aggregation occurs on a longer spatial scale than dispersal.

In the mathematics literature, the aggregation equations which have
similar forms to \eqref{eq1} have been studied extensively
\cite{BeLa07} \cite{BoVe05} \cite{BoVe06} \cite{BuDiF06}
\cite{BuCaMo07} \cite{La07} \cite{ToBe04}. The case of \eqref{eq1}
with $r=0$ and general choices of the kernel $K$ was considered by
Bodnar and Velazquez \cite{BoVe06}. There by an ODE argument the
authors proved the local well-posedness of \eqref{eq1}
\textit{without the density-dependent term} for $C^1$ initial data.
For a generic class of choices of the kernel $K$ and initial data,
they proved by comparing with a Burgers-like dynamics, the finite
time blowup of the $L_x^\infty$-norm of the solution. Burger and Di
Francesco \cite{BuDiF06} studied a class of one-dimensional
aggregation equations of the form
\begin{align*}
 \partial_t u = \partial_x \left( u \partial_x (a(u) -K*u +V ) \right), \quad\text{ in }\,
{(0,\infty)\times\R},
\end{align*}
where $V:\, \R\to \R$ is a given external potential and the nonlinear diffusion term $a(\rho)$ is assumed to be either $0$ or a strictly increasing function of $\rho$. In the case of no diffusion ($a\equiv 0$) they proved the existence of stationary solutions and investigated the weak
convergence of solutions toward the steady state. In the case of sufficiently small diffusion ($a(\rho)=\epsilon \rho^2$)
they proved the existence of stationary solutions with small support. Burger, Capasso and  Morale \cite{BuCaMo07}
studied the well-posedness of an equation similar to \eqref{eq1} but with a different diffusion term:
\begin{align*}
 \partial_t u +\nabla \cdot (u \nabla K* u) = div( u \nabla u),\quad \text{in}\;
  {(0,T)\times\R^d}.
  \end{align*}
For initial data $u_0 \in L^1_x(\R^d) \cap L^\infty_x(\R^d)$ with
$u_0^2 \in H^1_x(\R^d)$, they proved the existence of a weak
solution by using the standard Schauder's method. Moreover the
uniqueness of entropy solutions was also proved there. In connection
with the problem we study here, Laurent  \cite{La07} has studied in
detail the case of \eqref{eq1} \textit{without density-dependent
diffusion} (i.e. $r=0$ ) and proved local and global existence
results for a class of kernels $K$ with $H_x^s(s\ge 1)$ initial
data. More recently  Bertozzi and Laurent \cite{BeLa07} have
obtained finite-time blowup of solutions for the case of \eqref{eq1}
without diffusion (i.e. $r=0$) in $\R^d (d\ge 2)$ assuming compactly
supported radial initial data with highly localized support. 
Li and Rodrigo \cite{LR08} \cite{LiRod08prep} \cite{LiRod09}
studied the case of \eqref{eq1}
with fractional dissipation and proved finite-time blowup or global
wellposedness in various situations. We refer to \cite{LiRodZhang09} \cite{LiZhang09}
for the cases with singular kernels and further detailed studies
concerning sharp asymptotics and regularity of solutions.
We also mention that Bertozzi and Brandman \cite{BB08}
recently constructed $L_x^1\cap L_x^\infty$ weak solutions to \eqref{eq1} in $R^d$ ($d\ge 2$) and
with no dissipation ($r=0$) by following Yudovich's work on incompressible Euler equations
\cite{Yo63}. We refer the interested readers to \cite{PD84} \cite{HM89} \cite{I84} \cite{I85} \cite{IN87}
\cite{K78} \cite{MY82} \cite{NM83} and the references therein for some further rigorous studies.

From the analysis point of view, equation \eqref{eq1} is also connected with a general class of degenerate parabolic equations known as porous
medium equations, which takes the form
\begin{align} \label{eq_poro}
 \partial_t u= \partial_x ( u^m \partial_x u),\quad (t,x)\in  (0,\infty)\times\R,
\end{align}
where $m$ is a real number. These equations describe the ideal gas
flow through a homogeneous porous medium and other physical
phenomena in gas dynamics and plasma physics \cite{M37} \cite{ZR66}
\cite{Ar86}. Oleinik, Kalashnikov and Chzhou \cite{OKC58} proved the
existence and uniqueness of a global weak solution of
\eqref{eq_poro}. H\"older continuity of the weak solution is
established in \cite{Ar69} \cite{Kr69} \cite{Gi76} \cite{CF79}. As
for higher regularities, the Lipschitz continuity of the pressure
variable $v=\tfrac{m+1}m u$ is shown in \cite{Ar69} \cite{Di79}
\cite{Be83} \cite{AC86} \cite{CVW87}. Aronson and Vazquez
\cite{AV87} showed that $v$ becomes $C^\infty$ outside of the free
boundary after some waiting time. For derivative estimates of the
solution $u$, a local solution of \eqref{eq_poro} in
$W^{1,\infty}_x(\R)$ is constructed by Otani and Sugiyama
\cite{OS96}. In the case of \eqref{eq_poro} with $m$ being an even
natural number, Otani and Sugiyama \cite{OS01} proved the existence
of smooth solutions. We refer the interested readers to
\cite{AABL01} \cite{AG93} \cite{Be72} \cite{BCW00} \cite{BV05}
\cite{CC95} \cite{CF80} \cite{CV99} \cite{CVW87} \cite{CCV06}
\cite{CEL84} \cite{CL71} \cite{Ko99} \cite{Va07} and references
therein for more detailed studies and expositions.

Our starting point of the analysis of equation \eqref{eq1}
is to treat it as a degenerate parabolic problem with a nonlocal flux term. We
focus on constructing and analyzing classical solutions of \eqref{eq1}. The bulk of this
 paper is devoted to proving the existence and uniqueness of classical solutions to \eqref{eq1},
 which is done in section 2 and part of section 3. In the
final part of section 3, we prove the continuation and blowup criteria of solutions.
In the last section we prove that any smooth
initial data with compact support will lead to blowup of the gradient in finite time.
The analysis developed in this work can be extended to treat the $d$-dimensional ($d\ge 2$) 
case of \eqref{eq1} which we will address in a future publication.

We now state more precisely our main results. The first theorem establishes the existence of smooth local solutions for initial
data which is not necessarily nonnegative.

\begin{thm}[Existence and uniqueness of smooth local solution]  \label{thm1}
 Assume $u_0 \in \bigcap_{m=0}^\infty H^m_x(\mathbb R)$ ($u_0$ is not necessarily nonnegative). Then there exists a positive
$T_0 =T_0(\|u_0\|_{2}+\|\partial_x u_0 \|_{\infty})$ such that
\eqref{eq1} has a unique solution $u \in C^\infty([0,T_0]\times
\R)$. In particular $u \in C([0,T_0]\cap H^k_x)$ for any $k\ge 0$.
If $u_0\ge 0$, then $u(t)\ge 0$ for all $0\le t\le T_0$.
\end{thm}

\begin{rem}
The assumptions on the initial data $u_0$ in Theorem \ref{thm1} can
be weakened significantly (see for example Theorem \ref{thm_local}).
However in order to simplify the presentation, we do not state our
theorems here in its most general form.
\end{rem}

Our second theorem gives the blowup or continuation criteria of
solutions. Roughly speaking, it says that all the $L^p_x$-norm of
the solution cannot blow up and we can continue the solution as long
as we have a control of the gradient of the solution.

\begin{thm}[Blowup alternative and growth estimate of $L^p_x$-norm] \label{thm2}
Assume $u \in C^\infty([0,T)\times \R)$ is a maximal-lifespan
solution obtained in Theorem \ref{thm1}. Then either $T=+\infty$ or
$T<+\infty$ and
\begin{align*}
 \lim_{t\to T} \| \partial_x u(t) \|_{\infty} =+\infty.
\end{align*}
For any $2\le p<\infty$, there exists a generic constant $C$ such
that
\begin{align*}
 \| u(t) \|_{p} \le \|u_0 \|_{p} e^{Cpt}, \quad \forall\, t\in\ [0,T).
\end{align*}
If in addition $u_0\ge 0$, then $u(t) \ge 0$ for all $t\in [0,T)$,
and we also have the $p$-independent estimate for all $2\le p\le
+\infty$:
\begin{align*}
 \| u(t) \|_{p} \le \|u_0\|_{p} \exp \left ( Ct \|u_0\|_{2} e^{Ct} \right), \quad\forall\,
  t\in[ 0,T).
\end{align*}
In particular if $p=+\infty$, then
\begin{align*}
 \| u(t) \|_{\infty} \le \|u_0\|_{\infty} \exp\left( Ct \|u_0\|_{2} e^{Ct} \right),
 \quad\forall\, t\in [0,T).
\end{align*}
\end{thm}

The next theorem states that if we assume the initial data has a
little bit more integrability, then the local solution will inherit
this property. Note in particular that the $L^1_x$-norm of the solution is preserved for
all time if the initial data is nonnegative and in $L^1_x(\R)$.

\begin{thm}[Local solution in $L_x^p(\R)$] \label{thm3}
Assume $u_0 \in  \bigcap_{m=0}^\infty H^m_x(\mathbb R)$ and $u_0 \in
L_x^p(\R)$ for some $1\le p<2$. Then the local solution obtained in
Theorem \ref{thm1} also satisfies $u\in C([0,T], L_x^p)$. If in
addition $u_0 \ge 0$,  then $u(t)\ge 0$ for any $t\in [0,T]$. If
also $p=1$, then  $\|u(t)\|_{1} = \|u_0\|_{1}$, i.e. $L_x^1$-norm of
the solution is preserved.
\end{thm}


The last theorem states that any solution with smooth nonnegative initial data will blow up in finite time.

\begin{thm}[Finite time blowup]\label{thm4}
Let $u_0 \in C_c^\infty(\R)$ and $u_0 \ge 0$. Assume $u_0$ is not identically $0$.
Then there exists time $T<\infty$ such that the corresponding solution with $u_0$ as initial data blows up at time $T$
in the sense that
\begin{align} \label{eq_thm41}
 \lim_{t\to T} \| \partial_x u(t) \|_{\infty} =+\infty.
\end{align}
However all the $L^p_x$-norm of $u$ remain finite at the time of blowup, i.e.:
\begin{align} \label{eq_thm42}
 \sup_{0\le t < T} \| u(t) \|_{p} \le C(T)<\infty, \forall\; 1\le p \le \infty.
\end{align}
\end{thm}

\begin{rem}
In the case of \eqref{eq1} \textit{without diffusion} (i.e. $r=0$ in
\eqref{eq1}), if we take the kernel $K(x)=e^{-|x|}$, then the result
of Bodnar and Velazquez \cite{BoVe06} says that any solution with
initial data $u_0$ satisfying a slope condition will blow up in
finite time in the sense that $\|u(t)\|_{\infty}$ blows up. This is
highly in contrast with our result here when the
diffusion term does not vanish. In this case 
the solution will blow up at the level of the gradient, i.e. we have  $\|\nabla u(t)\|_{\infty}$ tends to infinity while all the
other $L^p_x$-norm  remain finite as $t\to T$, where $T$ is the
blowup time.
\end{rem}

\textbf{Notations.} Throughout the paper we denote $L_x^p=L_x^p(\R)$
for $1\le p \le \infty$ as the usual Lebesgue space on $\R$. We also
write $\|\cdot\|_p = \|\cdot\|_{L_x^p}$.
 For $s>0$, $s$ being an integer and
$1\le p\le \infty$, $W^{s,p}_x = W^{s,p}_x(\R)$ denotes the usual Sobolev space
\begin{align*}
 W^{s,p}_x = \bigl\{ f \in S^\prime(\R): \, \| f\|_{W^{s,p}}
  = \sum_{0\le j\le s} \| \partial_x^j f \|_{p} <\infty \bigr\}.
\end{align*}
When $p=2$, we denote $H^m_x=H^m_x(\R) = W^{2,p}_x(\R)$ and
$\|\cdot\|_{H^m}$ as its norm. Occasionally we shall use the
Sobolev space of fractional power $H_x^s(\R)$ whose norm can be
defined via Fourier transform:
$$
\|f\|_{ H^s}=\| (1+|\xi|)^s \hat f(\xi)\|_{L_{\xi}^2}.
$$
For any two quantities $X$ and $Y$, we use $X \lesssim Y$ or $Y
\gtrsim X$ whenever $X \leq CY$ for some constant $C>0$. A constant
$C$ with subscripts implies the dependence on these parameters.
We write $A=A(B_1,\cdots,B_k)$ when we want to stress that a quantity
$A$ depends on the quantities $B_1,\cdots,B_k$. 

 From now on we assume $r=1$ in \eqref{eq1} without loss of generality.
Same results hold for any $r>0$.

\section{The regularized equation and its wellposedness}
Since \eqref{eq1} is a degenerate parabolic equation, in order to
construct a local solution, we have to regularize the equation. To
this end, we consider the following regularized version of
\eqref{eq1}
\begin{align} \label{eq_WSR_1}
 \begin{cases}
  \partial_t u = (u^2+\epsilon) \partial_{xx} u + 2 u (\partial_x u)^2 - \partial_x(u\, K*\partial_x u ), \\
  u(0,x)=u_0(x).
 \end{cases}
\end{align}
Here $\epsilon>0$ is a parameter. We are going to prove the following

\begin{prop}[Local solution of the regularized equation]
\label{prop1}

 Assume $u_0 \in \bigcap_{m=0}^\infty H^m_x(\mathbb R)$ ($u_0$ is not necessarily nonnegative).
Then there exists a positive $T_0 =T_0(\|\partial_x u_0 \|_{\infty}
+\|u_0\|_{2})$ ($T_0$ is independent of $\epsilon$) such that
\eqref{eq_WSR_1} has a unique solution $u^\epsilon \in
C^\infty([0,T_0]\times \R)$ for any $\epsilon>0$. In particular
$u^\epsilon \in C([0,T_0];\, H^k_x)$ for any $k\ge 0$.
\end{prop}

\begin{cor} \label{cor1}
Assume $u_0 \in \bigcap_{m=0}^\infty H^m_x(\mathbb R)$ and $u_0 \in
L_x^p(\R)$ for some $1\le p<2$. Then the local solution obtained in
Proposition \ref{prop1} also satisfies $u^{\epsilon} \in
C([0,T_0],L_x^p)$. And $T_0$ can be chosen to be
$T_0=T_0(\|u_0\|_{p}+\|\partial_x u_0\|_{\infty})$.
\end{cor}

In the case $u_0 \in \bigcap_{m=0}^\infty H^m_x(\mathbb R) $ with $u_0 \ge 0$. The regularized equation admits
a global solution.
\begin{prop} \label{prop2}
 Assume $u_0 \in \bigcap_{m=0}^\infty H^m_x(\mathbb R)$ with $u_0 \ge 0$.
Then for any $T>0$  \eqref{eq_WSR_1} has a unique solution
$u^\epsilon \in  C^\infty([0,T]\times \R) $. In particular
$u^\epsilon \in C([0,T]; \, H^k_x)$ for any $k\ge 0$. Also
$u^\epsilon(t)\ge 0$ for any $t\ge 0$. If in addition $u_0 \in
L_x^p(\R)$ for some $1\le p<2$, then $u^\epsilon \in C([0,T];\,
L_x^p)$.
\end{prop}

\subsection{Proof of Proposition \ref{prop1}}
Our proof of Proposition \ref{prop1} is reminiscent of the
$L_x^\infty$ energy method used by Otani and Sugiyama \cite{OS01}
where they dealt with the porous medium equation \eqref{eq_poro}.
Denote the set
\begin{align*}
 B^k_{T} = \bigl\{ v\in C([0,T]; \, H^{2k+1}_x(\R)):\; \partial_{xx} v, \partial_t v \in L^2([0,T];\,H^{2k}_x(\R)) \bigr\}.
\end{align*}
As a very first step, we shall show the local existence of the
solution to \eqref{eq_WSR_1} in $B^k_T$. At this point, we need the
following lemma from \cite{OS01}.

\begin{lem} \label{lem518}
Consider the initial value problem
\begin{align}
 \begin{cases}
  \partial_t u = (u^2 +\epsilon) \partial_{xx} u + h(t,x), \\
  u(0) =u_0,
 \end{cases}
\end{align}
where $h \in L^2(0,T; \, H^{2k}_x)$, $u_0 \in H^{2k+1}_x(\R)$, there
exists a unique function $u=u(t,x) \in B^k_T$ which satisfies, for some constant
$C_1=C_1(\| u_0\|_{H^{2k+1}_x}, \| h\|_{L_t^2 H_x^{2k}},
k,\epsilon)$, that
\begin{align*}
 \| u\|_{L_t^\infty H^{2k+1}_x} + \epsilon^{\frac 12} \| u\|_{L_t^2 H^{2k+2}_x} \le C_1.
\end{align*}
For any $h_1$, $h_2 \in K_R^T=\bigl\{v:\, \| v\|_{L^2(0,T;\, H^{2k}_x)} \le R \bigr\}$, there exists
a constant $C_2=C_2(R,k,\epsilon)$ such that
\begin{align*}
 \| u_1 - u_2 \|_{L_t^\infty H^{2k+1}_x} + \epsilon^{\frac 12} \| u_1 -u_2 \|_{L_t^2 H^{2k+2}_x}
\le C_2 e^{C_2 T} \| h_1 -h_2 \|_{L_t^2 H^{2k}_x}.
\end{align*}

\end{lem}

\begin{proof}
 See \cite{OS01}.
\end{proof}

Based on this lemma, we establish the local existence
of solutions of the regularized equation. Note however that the time
of existence of the solution depends on the regularization parameter
$\epsilon$. We have the following

\begin{lem}[Local existence of the regularized equation] \label{lem_local_exist}
 Let $u_0 \in H^{2k+1}_x(\R)$, $k\ge 2$. There exists $T_0=T_0(\|u_0\|_{H^{2k+1}},k, \epsilon)>0$, such that
\eqref{eq_WSR_1} has a unique solution $u \in B^k_{T_0}$.
\end{lem}
\begin{proof}
 Introduce the auxiliary equation
\begin{align} \label{eq538}
 \begin{cases}
  \partial_t u = (u^2+\epsilon)\partial_{xx} u + h(t,x) \\
 u(0)=u_0
 \end{cases}
\end{align}
By Lemma \ref{lem518}, for any $u_0 \in H^{2k+1}_x(\R)$, $h\in
L^2(0,T;H^{2k}_x)$, \eqref{eq538} has a unique solution $u_h \in
B^k_T$. Next define the map
\begin{align*}
 \phi(h) = 2 u_{h} (\partial_x u_h)^2 - \partial_x ( u_h K*\partial_x u_h).
\end{align*}
It suffices to show that for some suitable $R$ and $T_0$, $\phi$ is a contraction from the set
$K^{T_0}_R = \bigl\{ v\in L^2(0,T_0;\, H^{2k}_x);\, \| v\|_{L^2(0,T_0;H^{2k}_x)} \le R \bigr\}$ into itself.
First we show that $\phi$ maps $K_R^{T_0}$ into $K_R^{T_0}$. Take any $h\in K_R^T$, by the fact that $H^{s}_x(\R)$ is an
algebra when $s>\frac 12$, we have
\begin{align*}
 \| \phi(h) \|_{H^{2k}} & \le \|u_h \|_{H^{2k}} \| u_h \|_{H^{2k+1}}^2 + \| \partial_x K* \partial_x u_h u_h\|_{H_x^{2k}}
+ \| K* \partial_x u_h \partial_x u_h \|_{H^{2k}} \\
& \le C (\|u_h \|_{H^{2k+1}}^3+\|u_h\|_{H^{2k+1}}^2),
\end{align*}
where $C$ is a generic constant. Then Lemma \ref{lem518} implies that
\begin{align*}
 \| \phi(h) \|_{L^2(0,T;\, H^{2k}_x)} &\le CT^{\frac 12} (\|u_h \|_{L_t^{\infty}(0,T;H^{2k+1})}^3
 +\|u_h \|_{L_t^{\infty}(0,T;H^{2k+1})}^2)\\
& \le T^{\frac 12} C_3(\|u_0\|_{H_x^{2k+1}}, R,k,\epsilon).
\end{align*}
This shows  that if we take $T \le C_3^2/R^2$, then $\phi$ maps $K_R^T$ into $K_R^T$.
Similarly we can show that $\phi$ is Lipschitz, i.e. for any $h_1$, $h_2 \in K_R^T$, we have
\begin{align*}
 \| \phi(h_1) - \phi(h_2) \|_{L^2(0,T;\, H^{2k}_x)} \le C_4 T^{\frac 12} e^{C_4 T}
\| h_1- h_2\|_{L^2(0,T;\, H^{2k}_x)},
\end{align*}
where $C_4=C_4(\|u_0\|_{H^{2k+1}}, R,k,\epsilon)$. Therefore if $T$
is sufficiently small, then $\phi$ is a contraction on $K_R^T$. The
lemma is proved.
\end{proof}

Lemma \ref{lem_local_exist} is not satisfactory since the time of
existence of the solution depends on the regularization parameter
$\epsilon$. The following lemma removes this
dependence, and at the same time, weakens the dependence on the
initial norm.

\begin{lem} \label{lem301}
 Let $u_0 \in H^{2k+1}_x(\R)$, $k\ge 2$. Then there exists $T_0=T_0(\|\partial_x u_0\|_{\infty}
  +\|u_0\|_2)$ such that
for any $\epsilon>0$, \eqref{eq_WSR_1} has a unique solution $u^\epsilon \in B^k_{T_0}$.
 If $u_0\ge 0$, then $u^\epsilon(t)\ge 0$
for any $0\le t\le T_0$.
\end{lem}
\begin{proof}
By Lemma \ref{lem_local_exist}, we can continue the solution as long
as we can control the ${H_x^{2k+1}}$-norm of $u^\eps$. 
In the following we give the a priori control of $u^\epsilon$. As we
will see, the $H^{2k+1}_x$-norm of $u^{\eps}$ will stay bounded on a time
interval $[0,T_0]$ for a certain small $T_0$ depending only on the norms $\|\partial
_xu_0\|_{\infty}$, $\|u_0\|_2$. We
are then left with the task of estimating the various Sobolev norms which we
will do in several steps. For simplicity of notations we shall write
$u^\epsilon$ as $u$ throughout this proof.


\texttt{Step 1}: $L_x^p$-norm estimate for $2\le p \le +\infty$.
Multiply both sides of \eqref{eq_WSR_1} by $|u|^{p-2} u$ and
integrate over $x$, we have
\begin{align} \label{eq1144}
 \frac 1 p \frac d {dt} \| u(t) \|_p^p & =-\epsilon(p-1) \int |u|^{p-2} u_x^2 dx  -(p-1) \int u_x^2 |u|^p dx \notag \\
&\qquad+ (p-1) \int (\partial_x K*  u) u |u|^{p-2} \partial_x u dx \\
&\le -(p-1) \int |u|^{p-2} \left( uu_x - \frac {\partial_x K*u} 2 \right)^2 dx \notag \\
&\qquad+ \frac {p-1} 4 \int |u|^{p-2} (\partial_x K*u )^2 dx \notag\\
& \le  \frac {p-1} 4 \| \partial_x K * u \|_p^2 \| u\|_p^{p-2}  \notag\\
& \le \frac {p-1}4 \| \partial_x K \|_1^2 \| u\|_p^p. \notag
\end{align}
Gronwall's inequality implies that
\begin{align} \label{eq948}
 \| u(t) \|_p \le \|u_0\|_p \cdot \exp\left( \frac{p-1}4 \|\partial_x K\|_1^2 t \right), \quad\forall\, 2\le p<+\infty.
\end{align}
In particular we have
\begin{align}
 \| u(t) \|_2 \le \|u_0\|_2 \exp \left( \frac14 \|\partial_x K\|_1^2 t \right).
\end{align}
To control $L_x^\infty$-norm, we are going to choose $T_0<\frac 1
{\|\partial_x K\|_1}$. This implies that $\|u(t)\|_{2}$ has a
uniform bound on any such $[0,T_0]$. Now use again the RHS of
\eqref{eq1144} to get
\begin{align*}
 \frac 1 p \frac d {dt} \| u(t) \|_p^p & \le (p-1) \int (\partial_x K*  u) u |u|^{p-2} \partial_x u dx \\
& =\frac {p-1} p \int (\partial_x K* u) \partial_x (|u|^p) dx \\
&=-\frac {p-1} p \int (\partial_{xx}K*u) |u|^p dx.
\end{align*}
Now we use the fact that $K(x) =e^{-|x|}$ and 
therefore $\partial_{xx} K(x) = -2 \delta(x) +e^{-|x|}$, and this gives 
\begin{align*}
  \frac 1 p \frac d {dt} \| u(t) \|_p^p & \le -\frac {2(p-1)} p \int u |u|^p dx + \frac {p-1} p \int (e^{-|x|}*u) |u|^p dx \\
& \le \frac {2(p-1)} p \| u(t) \|_{p+1}^{p+1} + \frac {p-1} p \| u(t) \|_p^p \| e^{-|x|}\|_2 \| u\|_2 \\
&\lsm \frac {2(p-1)} p \| u(t) \|_{p+1}^{p+1} + \frac {p-1} p \| u(t) \|_p^p \\
& \lsm \frac {2(p-1)} p \| u\|_p^p \|u\|_\infty + \frac {p-1} p \| u\|_p^p.
\end{align*}
This implies
\begin{align*}
 \frac d {dt} \|u(t)\|_p \lsm \frac {2(p-1)} p \|u\|_p \|u\|_{\infty} + \frac {p-1} p \| u(t) \|_p.
\end{align*}
Integrating over $[0,t]$, we obtain
\begin{align*}
 \| u(t) \|_p \le \|u_0\|_p + C\int_0^t (1+\|u(s)\|_{\infty}) \|u(s)\|_p ds.
\end{align*}
Letting $p\to \infty$, we obtain
\begin{align*}
 \|u(t) \|_\infty \le \|u_0\|_\infty + C \int_0^t (1+\|u(s)\|_\infty) \|u(s)\|_\infty ds.
\end{align*}
Now it's easy to see that if we choose $T_0$ sufficiently small
depending on $\|u_0\|_2,\|u_0\|_\infty$, then we have
\begin{align}\label{ru}
 \sup_{0\le t\le T_0} \|u(t)\|_{r} \le C \| u_0 \|_r, \quad \forall\, 2\le r\le \infty.
\end{align}

\texttt{Step 2}: Control of $\|\partial_x u(t) \|_\infty$. At this point, one can appeal directly to the maximum principle
(see for example Theorem 11.16 of \cite{Lie96}). We give a rather direct estimate here without using the maximum principle.
Differentiating both sides of \eqref{eq_WSR_1} w.r.t $x$ and denoting $v=u_x$, we have
\begin{align} \label{eq1242}
 \partial_t v = \epsilon v_{xx}+u^2 v_{xx} +6 u vv_x +2 v^3 +G,
\end{align}
 where
\begin{align*}
G= \partial_{xx} (uK*v ).
\end{align*}
Multiplying both sides of \eqref{eq1242} by $|v|^{p-2}v$ and integrating over $x$, we obtain
\begin{align} \label{eq1247}
 \frac 1 p \frac d {dt} \| v(t) \|_p^p & = \int u^2 v_{xx} |v|^{p-2} v dx + \epsilon \int v_{xx} |v|^{p-2} vdx +6 \int u v_x |v|^p dx \notag\\
&\quad + 2 \int|v|^{p+2} dx + \int G(x) |v|^{p-2} vdx.
\end{align}
Integrate by parts and we have
\begin{align*}
 \int u^2 v_{xx} |v|^{p-2} v dx &= -2 \int u v_x |v|^p dx - (p-1) \int u^2 v_x^2 |v|^{p-2} dx,\\
\int u v_x |v|^p dx &= - \frac 1{1+p} \int |v|^{p+2} dx \le 0,\\
\epsilon \int v_{xx} |v|^{p-2} vdx &= -\epsilon(p-1) \int |v|^{p-2} v_x^2 dx \le 0.
\end{align*}
Plugging the above estimates into \eqref{eq1247}, we get
\begin{align*}
 \frac 1p \frac d{dt} \|v(t)\|_p^p \le 2 \int |v|^{p+2} dx + \int G(x) |v|^{p-2} v dx.
\end{align*}
Now we compute
\begin{align*}
 \int G |v|^{p-2} v dx =& \int \partial_{xx} (K*v u) |v|^{p-2} v dx \\
=& \int \partial_{xx}K* v \, u |v|^{p-2} v dx + 2 \int \partial_x K*v |v|^p dx + \int K*v v_x |v|^{p-2} v dx \\
\lsm & \int |u| |v|^p dx + \int( e^{-|x|}* v) u |v|^{p-2} v dx + \| v\|_p^p \|v \|_\infty\\
\lsm &\|v\|_p^p (\| u\|_\infty + \| v\|_\infty).
\end{align*}
Finally we obtain
\begin{align*}
 \frac 1 p \frac d {dt} \| v(t) \|_p^ p & \lsm \|v \|_{p+2}^{p+2} + \| v\|_p^p ( \| u\|_\infty + \|v\|_\infty) \\
& \lsm \| v\|_p^p ( \|v\|_\infty^2 + \|v\|_\infty + \|u\|_\infty).
\end{align*}
Integrating over $t$, letting $p\to \infty$, and recalling $v=u_x$, we have
\begin{align}\label{rux}
 \sup_{0\le t\le T_0} \|\partial_x u(t) \|_\infty & \le C \| \partial_x
 u_0\|_\infty,
\end{align}
for sufficiently small $T_0=T_0(\|u_0\|_\infty, \|u_0\|_2, \| \partial_x u_0 \|_\infty)$.

\texttt{Step 3:} Control of $\| \partial_{xx} u(t) \|_{p}$ for $2\le
p\le +\infty$. First take $ 2\le p <+\infty$ and compute
\begin{align*}
 &\frac 1 p \frac d {dt} \int |\partial_{xx} u |^p dx  \\
 = & \frac 13 \int |\partial_{xx} u|^{p-2} \partial_{xx} u \cdot (u^3)_{xxxx} dx
+ \epsilon \int |\partial_{xx} u|^{p-2} \partial_{xx} u u_{xxxx} dx \\
&\quad -\int |\partial_{xx} u|^{p-2} \partial_{xx} u \partial_{xxx}((\partial_x K*u) u) dx.
\end{align*}
We further estimate by using integration by parts,
\begin{align} \label{eqMay222}
 & \int |\partial_{xx} u|^{p-2} \partial_{xx} u \cdot (u^3)_{xxxx} dx \notag\\
= & 3 \int |\partial_{xx} u|^{p-2} \partial_{xx} u \cdot u^2 \partial_{xxxx} udx +24 \int |\partial_{xx} u|^{p-2}
\partial_{xx} u\cdot u\partial_x u \partial_{xxx} u dx \notag\\
&\quad +18\int |\partial_{xx} u|^p u \partial_{xx} u dx + 36\int |\partial_{xx} u|^p (\partial_x u)^2 dx \notag\\
= & - 3(p-1)\int (\partial_{xxx} u)^2 |\partial_{xx} u|^{p-2} u^2 dx +18(1-p)\int |\partial_{xx} u|^{p-2}
\partial_{xx} u\cdot u\partial_x u \partial_{xxx} u dx \notag\\
&\quad +18\int |\partial_{xx} u|^p (\partial_x u)^2 dx.
\end{align}
Using Cauchy Schwartz inequality:
$$
ab\le \frac 16 a^2+\frac 32 b^2,
$$
we bound the second term in \eqref{eqMay222} as
\begin{align*}
|18(1-p)\int &|u_{xx}|^{p-2}\partial_{xx} u u\partial_x
u\partial_{xxx}u dx| \\
&\le 27(p-1)\int|u_{xx}|^p|\partial_x u|^2 dx+3(p-1)\int
|u_{xx}|^{p-2} u^2|\partial_{xxx} u|^2 dx.
\end{align*}
Hence,
\begin{align*}
\int |\partial_{xx} u|^{p-2} \partial_{xx} u \cdot (u^3)_{xxxx} dx
\le & (27p-9)\int |\partial_{xx} u|^p (\partial_x u)^2 dx.
\end{align*}
Also it is obvious that
\begin{align*}
 &\epsilon \int |\partial_{xx} u|^{p-2} \partial_{xx} u u_{xxxx} dx \\
= & -\eps p \int (\partial_{xxx} u)^2 | \partial_{xx} u|^{p-2} dx
\le 0.
\end{align*}
Similarly we have
\begin{align*}
 & \int |\partial_{xx} u|^{p-2} \partial_{xx} u \cdot \partial_{xxx} ( (\partial_x K*u) u) dx \\
\lsm & \int |\partial_{xx} u|^{p-2} \partial_{xx} u \partial_{xxx} u (\partial_x K* u) dx + \int| \partial_{xx} u|^{p} \partial_x K* \partial_x udx \\
&\quad + \int |\partial_{xx} u|^{p-2} \partial_{xx} u \partial_x u \partial_x K* \partial_{xx} u dx + \int |\partial_{xx} u|^{p-2} \partial_{xx} u
u \partial_{xx} K* \partial_{xx} u dx \\
\lsm & \int |\partial_{xx} u|^p \partial_x K* \partial_x u dx + \int |\partial_{xx} u|^{p-2} \partial_{xx} u \partial_x u \partial_x K*
\partial_{xx} udx\\
&\quad + \int |\partial_{xx} u|^p u dx + \int |\partial_{xx} u|^{p-2} \partial_{xx} u u (e^{-|x|}* \partial_{xx} u) dx \\
\lsm & \| \partial_{xx} u \|_p^p \cdot ( \| \partial_x u\|_\infty + \| u \|_\infty).
\end{align*}
Collecting all the estimates and we have
\begin{align*}
 \frac 1 p \frac d {dt} \| \partial_{xx} u(t) \|_p^p & \lsm
\| \partial_{xx} u\|_{p}^{p} \| \partial_x u\|_\infty^2+ \| \partial_{xx} u \|_p^p ( \|\partial_x u\|_\infty^2
+\| \partial_x u\|_\infty+ \| u\|_\infty) \\
& \lsm \| \partial_{xx} u\|_p^p ( \|\partial_x u\|_\infty^2
+\| \partial_x u\|_\infty+ \| u\|_\infty).
\end{align*}
Integrating over $t$  we have
\begin{align}\label{ruxx}
 \sup_{0\le t\le T_0} \|\partial_{xx} u(t) \|_p  \le e^{CpT_0} \| \partial_{xx} u_0\|_p.
\end{align}
To estimate $\| u_{xx} \|_\infty$ we need to estimate $\| u_{xxx} \|_2$. We have
\begin{align*}
 \frac d {dt} \| u_{xxx}(t) \|_2^2 & \lsm- \int u_{xxxx} (u^3)_{xxxx} dx +\epsilon \int u_{xxxxx} u_{xxx} dx
- \int \partial_{xxxx}(\partial_x K*u u) u_{xxx} dx \\
& \lsm -\int u_{xxxx}^2 u^2 dx + \int u_{xxxx} (uu_x u_{xxx} + u_{xx}u_x^2 + u_{xx}^2 u) dx \\
 &\quad- \int \partial_{xxxx}(\partial_x K*u u) u_{xxx} dx \\
& \lsm \int (u_{xxx}^2 u_x^2 + u_{xx}^4 + u_{xx}^2 u_x^2) dx- \int \partial_{xxxx}(\partial_x K*u u) u_{xxx} dx.
\end{align*}
Since
\begin{align*}
 & \int \partial_{xxxx}( (\partial_x K*u) u) u_{xxx} dx \\
\lsm & \int (\partial_x K*\partial_x u) u_{xxx}^2 dx + \int (\partial_x K*\partial_{xx} u) \partial_{xx} u \partial_{xxx} u dx \\
&\quad + \int (\partial_x K*\partial_{xxx}u) uu_{xxx} dx + \int u (\partial_{xx}K*\partial_{xxx}u) u_{xxx} dx \\
\lsm & \| u_{xxx} \|_2^2 \|\partial_x u\|_\infty+ \| u_{xxx}\|_2 \| \partial_{xx} u\|_2^2 + \| \partial_{xxx} u\|_2^2 \| u\|_\infty.
\end{align*}
Finally we get
\begin{align*}
 \frac d {dt} \| u_{xxx}(t) \|_2^2 &\lsm \| u_{xxx}\|_2^2 ( \|u_x\|_2^2 + \| u_x\|_\infty +\|u\|_\infty) \\
&\quad + \|u_{xxx}\|_2 \| \partial_{xx} u\|_2^2 + \| u_{xx} \|_4^4 + \| u_{xx} \|_2^2 \| u_x\|_\infty^2.
\end{align*}
Gronwall then implies that
\begin{align}
 \sup_{0\le t\le T} \|\partial_{xxx} u(t) \|_2 \le C \| \partial_{xxx}
 u_0\|_2.\label{three grad}
\end{align}
Since $\|u_{xx}\|_\infty \lsm \|u\|_2^{\frac 16} \|\partial_{xxx} u\|_2^{\frac 56}$, we conclude that
\begin{align*}
 \sup_{0\le t \le T_0} \| u_{xx}(t)\|_{\infty} \le C.
\end{align*}

\texttt{Step 4}: Control of $\|\partial_x^{2k+1} u\|_2$. We compute
\begin{align*}
 \frac d {dt} \| \partial_x^{2k+1} u\|_2^2 & = \frac 13 \int \partial_x^{2k+3} (u^3) \partial_x^{2k+1} udx
+ \epsilon \int \partial_x^{2k+3} u \partial_x^{2k+1} udx \\
&\quad -\int \partial_x^{2k+2} (K*\partial_x u u) \partial_x^{2k+1} u dx \\
&=: I+II+III
\end{align*}
For $I$, we integrate by parts and obtain 
\begin{align}
I&= -\frac 13 \int \partial_x^{2k+2}(u^3) \partial_x^{2k+2} u dx \notag\\
& =-\int u^2(\partial_x^{2k+2}u)^2 dx-C\int uu_x\partial_x^{2k+1}u
\partial_x^{2k+2}u dx\notag\\
&\qquad +\sum_{\substack{l+m+n= 2k+2\\l\le m\le n\le
2k}}C_{l,m,n}\int\partial_x^l u\partial_x^m u\partial_x^n
u\partial_x^{2k+2} udx\notag\\
&=-\int u^2(\partial_x^{2k+2}u)^2 dx+\frac C
2\int(u_x^2+uu_{xx})(\partial_x^{2k+1}u)^2 dx\label{049}\\
&\qquad+\sum_{\substack{l+m+n=2k+3\\l\le m\le n\le 2k+1}}C_{l,m,n}
\int\partial_x ^l u\partial_x^m u\partial x^n u\partial_x^{2k+1}u
dx\label{050}
\end{align}
Discarding the negative term and using H\"older's inequality, we
bound \eqref{049} as
\begin{equation}
\eqref{049}\lsm
(\|u_x\|_{\infty}^2+\|u\|_{\infty}\|u_{xx}\|_{\infty})\|\partial_x^{2k+1}
u\|_2^2.
\end{equation}
To estimate the summand in \eqref{049}, we discuss several cases.

Case 1. $n=2k+1$. In view of the constraint, we have $(l,m)=(0,2)$
or $ (1,1)$. H\"older's inequality then gives that
\begin{align*}
|\int \partial_x^l u\partial_x^m u\partial_x^n u\partial_x^{2k+1} u
dx|\lsm
(\|u_x\|_{\infty}^2+\|u\|_{\infty}\|u_{xx}\|_{\infty})\|\partial_x^{2k+1}
u\|_2^2.
\end{align*}

Case 2. $n\le 2k$, $l=0$. Since $l+m+n=2k+3$ we must have $m\ge 3$.
We choose $2<p,q<\infty$ such that $\frac 1p+\frac 1q=\frac 12$ and
use H\"older to get
\begin{align*}
|\int \partial_x^l u\partial_x^m u\partial_x^n u\partial_x^{2k+1} u
dx|\lsm \|u\|_{\infty}\|\partial_x^m u\|_p\|\partial_x^n
u\|_q\|\partial_x^{2k+1} u\|_2.
\end{align*}
Using the interpolation inequality
$$
\|\partial_x^m u\|_p\le \|\partial_x^{2k+1}
u\|_2^{\theta}\|u\|_{\dot H^{\frac 52}}^{1-\theta},\
\theta=\frac{m-2-\frac 1p}{2k-\frac 32},
$$
we have
\begin{align*}
|\int \partial_x^l u\partial_x^m u\partial_x^n u\partial_x^{2k+1} u
dx|&\le \|u\|_{\infty}\|\partial_x^{2k+1} u\|_2^2\|u\|_{\dot
H^{\frac
52}}\\
&\le \|u\|_{\infty}\|u\|_{H^3}\|\partial_x^{2k+1} u\|_2^2.
\end{align*}

Case 3. $n\le 2k$, $l=1$. In this case $m,n\ge 2$, we can choose
$2<p,q<\infty $ such that $\frac 1p+\frac 1q=\frac 12$ and use the
interpolation inequality
\begin{equation}\label{inter}
\|\partial_x^m u\|_p\le \|\partial_x^{2k+1} u\|_2^{\frac{m-1-\frac
1p}{2k-\frac 12}}\|u\|_{\dot H^{\frac 32}}^{1-\frac{m-1-\frac
1p}{2k-\frac 12}}
\end{equation}
 to get
\begin{align*}
|\int \partial_x^l u\partial_x^m u\partial_x^n u\partial_x^{2k+1} u
dx|&\le \|u_x\|_{\infty}\|\partial_x^m u\|_p\|\partial_x ^n
u\|_q\|\partial_x^{2k+1} u\|_2\\
&\le \|u_x\|_{\infty}\|\partial_x^{2k+1} u\|_2^2\|u\|_{\dot
H^{\frac 32}}\\
&\le \|u_x\|_{\infty}\|u\|_{H^2}\|\partial_x^{2k+1} u\|_2^2.
\end{align*}
Case 4. $n\le 2k$, $l,m,n\ge 2$. Choosing $2<p,q,r<\infty$ such that
$\frac 1p+\frac 1q+\frac 1r=\frac 12$ and using the interpolation
inequality \eqref{inter}, we get
\begin{align*}
|\int \partial_x^l u\partial_x^m u\partial_x^n u\partial_x^{2k+1} u
dx|&\le \|\partial_x^l u\|_p\|\partial_x^m u\|_q\|\partial_x^n
u\|_r\|\partial_x^{2k+1} u\|_2\\
&\le \|u\|_{\dot H^{\frac 32}}^2 \|\partial_x^{2k+1} u\|_2^2\\
&\le \|u\|_{H^2}^2\|\partial_x^{2k+1} u\|_2^2.
\end{align*}
Collecting all the estimates together, we conclude
\begin{align*}
I\le &(\|\partial_x u\|_{\infty}^2
+\|u\|_{\infty}\|u_{xx}\|_{\infty}
+\|u\|_{\infty}\|u\|_{H^3}\\
&+\|u_x\|_{\infty}\|u\|_{H^2}+\|u\|_{H^2}^2)\|\partial_x^{2k+1}
u\|_2^2.
\end{align*}

For term $II$, we simply have
\begin{align*}
 II = -\epsilon \int (\partial_x^{2k+2} u)^2 dx \le 0.
\end{align*}
For III, we compute:
\begin{align*}
 III & = - \int \partial_x^{2k+2} (K*\partial_x u) u \partial_x^{2k+1} u dx - \int K*\partial_x u \partial_x^{2k+2} u \partial_x^{2k+1} udx \\
&\quad + \sum_{\substack{l+m= 2k+2 \\l,m\le 2k+1}} C_{l,m}\int \partial_x^l(K*\partial_x u) \partial_x^m u \partial_x^{2k+1} u dx \\
& = - \int \partial_{xx} K* \partial_x^{2k+1} u u \partial_x^{2k+1} u dx + \frac 12 \int \partial_x K*\partial_x u (\partial_x^{2k+1} u)^2 dx \\
&\quad + \sum_{\substack{l+m= 2k+2 \\l,m\le 2k+1}}C_{l,m}
 \int \partial_x K* \partial_x^l u \partial_x^m u \partial_x^{2k+1} u dx \\
& =2\int u(\partial_x^{2k+1} u)^2 dx -\int (e^{-|x|}*
\partial_x^{2k+1} u)
 \partial_x^{2k+1} u u dx +\frac 12\int(\partial_x K*\partial_x u)(\partial_x^{2k+1} u)^2 dx\\
&\quad + \sum_{\substack{l+m\le 2k+2 \\l,m\le 2k+1}}
C_{l,m}\int \partial_x K* \partial_x^l u \partial_x^m u \partial_x^{2k+1} u dx \\
& \lsm (\| u\|_\infty +\|\partial_x u\|_{2})\|
\partial_x^{2k+1} u\|_2^2 \\
&\qquad+ \sum_{\substack{l+m=2k+2 \\l,m\le 2k+1}}C_{l,m} \|
\partial_x^{2k+1} u\|_2 \| \partial_x^l u\|_2
\| \partial_x^m u\|_2 .\\
\end{align*}
Using the interpolation inequality
$$
\|\partial_x^l u\|_2\le \|\partial_x^{2k+1}
u\|_2^{\frac{l-1}{2k}}\|u_x\|_2^{1-\frac{l-1}{2k}},
$$
we finally get
$$
III\lsm (\|u\|_{\infty}+\|u_x\|_{2})\|\partial_x^{2k+1} u\|_2^2.
$$
Summarizing the estimates of $I,II,III$ we obtain
\begin{align*}
\frac d{dt}\|\partial_x^{2k+1} u(t)\|_2^2& \lsm
\|\partial_x^{2k+1}u\|_2^2(\|u_x\|_{\infty}^2+\|u\|_{\infty}\|u_{xx}\|_{\infty}
\\
&\qquad+\|u\|_{\infty}\|u\|_{H^3}+\|u_x\|_{\infty}\|u\|_{H^2}+\|u\|_{\infty}+\|u_x\|_{2})\\
&\lsm \|\partial_x^{2k+1} u\|_2^2(\|u\|_{H^3}+\|u\|_{H^3}^2).
\end{align*}
Using Gronwall and \eqref{ru}, \eqref{ruxx}, \eqref{three grad} we
get
\begin{align*}
 \sup_{0\le t\le T_0} \| \partial^{2k+1}_x u(t) \|_2 \lsm e^{CT_0} \| \partial^{2k+1}_x
 u_0\|_2,
\end{align*}
for sufficiently small $T_0=T_0(\|u_0\|_2,\|\partial_x u_0\|_2)$. 

This concludes the estimate of the $H_x^{2k+1}$-norm of $u^\epsilon$.
Finally if $u_0\ge 0$, then by the weak maximum principle we have $u^\epsilon(t)
\ge 0$ for any $0\le t\le T_0$.  The lemma is proved.
\end{proof}
We are now ready to complete the
\begin{proof}[Proof of Proposition \ref{prop1}]
This follows directly from Lemma \ref{lem_local_exist} and \ref{lem301}. In particular note that
in Lemma \ref{lem301} the time of existence of the local solution does not depend on $\epsilon$.
\end{proof}

\begin{proof}[Proof of Corollary \ref{cor1}]
Assume $u_0 \in \bigcap_{m=0}^\infty H_x^m(\R) \bigcap L_x^p(\R)$
for some $1\le p<2$.
Then using \eqref{eq_WSR_1} and Duhamel's formula, 
we can write $u^\epsilon(t)$ as
\begin{align*}
 u^\epsilon(t) &= e^{\epsilon t\partial_{xx}} u_0 + \int_0^t e^{\epsilon (t-s) \partial_{xx}}  
\left( \frac 13 \partial_{xx} ((u^\epsilon)^3) \right)(s) \\
 &\qquad - \int_0^t e^{\epsilon(t-s) \partial_{xx} } \partial_x ( u^\epsilon \,\partial_x K*u^\epsilon )(s)  ds.
\end{align*}
Since $\|e^{\epsilon t\partial_{xx}} f\|_p \lsm \|f \|_p$ for any $\epsilon>0$, we can estimate the $L^p_x$ norm of $u^\epsilon$ as
\begin{align*}
  \| u^\epsilon(t) \|_p &\lsm \| u_0\|_p +
 \int_0^t \left \| \frac 13 \partial_{xx}((u^\epsilon)^3) - \partial_x( (\partial_x K*u^\epsilon) u^\epsilon) \right\|_p ds \\
 & \lsm \| u_0\|_p + \int_0^t \| u^\epsilon(s) \|_{W^{2,3p}}^3+ \| u^\epsilon(s) \|_{W^{1,2p}}^2 ds \\
 & \lsm \| u_0\|_p + \int_0^t \| u^\epsilon(s)\|_{H^3}^3 + \|u^\epsilon(s) \|_{H^2}^2  ds,
\end{align*}
where the last inequality follows from the Sobolev embedding. This
estimate shows that $u^\epsilon(t) \in L_x^p$ for any $t$. The continuity (including right continuity at $t=0$) follows from similar estimates.
 We omit the details. Finally since for $1\le p<2$, $\|u_0\|_2 + \|\partial_x u_0\|_\infty \lsm \|u_0\|_p + \| \partial_x u_0\|_\infty$, one
 can choose the time interval sufficiently small depending only on $\|u_0\|_p + \| \partial_x u_0\|_\infty$.
\end{proof}


\subsection{Proof of Proposition \ref{prop2}}
In the case $u_0 \ge 0$, we shall show the regularized equation has
a global solution. The key point here is that by using
the positivity, we can obtain the apriori boundedness of the $L_x^p$-norm ($2\le p\le \infty$) of the solution 
on any finite time interval. To control the $L_x^{\infty}$-norm of the gradient, we need the following
lemma from \cite{Lie96}.

\begin{lem} \label{lemM}
Let $\Omega = (0,T) \times \R$ and $u \in C^{2,1}(\Omega) \cap C( \bar \Omega) $ be a solution
to the parabolic equation:
\begin{align*}
 \begin{cases}
  \partial_t u = b(t,x,u,u_x) u_{xx} +a(t,x,u,u_x), \quad (t,x) \in (0,T) \times \R \\
 u|_{t=0} = u_0
 \end{cases}
\end{align*}
with the following properties:
\begin{enumerate}
 \item Uniform parabolicity: there exists constants $\beta_0$, $\beta_1>0$ such that
\begin{align*}
 \beta_0 b(t,x,z,p) p^2 \ge |a(t,x,z,p)|,\quad\forall\, |p|\ge \beta_1.
\end{align*}

\item There is a constant $M>0$ such that
\begin{align*}
 |u(t,x) - u(t,y) | \le M,\quad \forall\, (t,x), (t,y) \in \Omega.
\end{align*}

\item There exists $L_0>0$ such that
\begin{align*}
 |u_0(x) -u_0(y) | \le L_0 |x-y|, \quad \forall\, x,y\in \R.
\end{align*}
\end{enumerate}
Under all the above assumptions, we have the following bound:
\begin{align*}
 \sup_{0\le t\le T} \| u_x(t) \|_\infty \le 2 (L_0+\beta_1) e^{\beta_0 M}.
\end{align*}

\end{lem}

\begin{proof}
See Lemma 11.16 of \cite{Lie96}.
\end{proof}

Incorporating Lemma \ref{lemM} with the idea of the
proof of Lemma \ref{lem301}, we get the following global analogue of
Lemma \ref{lem301}.
\begin{lem} \label{lem301B}
 Let $u_0 \in H^{2k+1}_x(\R)$, $k\ge 2$ and $u_0 \ge 0$. Then for any $T>0$ and
for any $\epsilon>0$, \eqref{eq_WSR_1} has a unique solution $u^\epsilon \in B^k_{T}$. Also
$u^\epsilon(t)\ge 0$ for any $t\ge 0$.
\end{lem}
\begin{proof}
The positivity of $u^\epsilon$ follows easily from the weak maximum
principle. By Lemma \ref{lem_local_exist} we only need to get an
apriori control of the $H^{2k+1}_x$-norm of $u$ on an arbitrary time
interval $[0,T]$.

\texttt{Step 1}: $L_x^p$-norm control for $2\le p\le +\infty$. By
the same estimates as in Step 1 of the proof of Lemma \ref{lem301} we have the
following a priori $L_x^2$-norm estimate:
\begin{align*}
 \|u(t) \|_2 \le \| u_0\|_2 \exp\left(\frac 14 \| \partial_x K\|_1^2 t \right).
\end{align*}
To obtain the $L_x^\infty$-norm estimate, we take $2<p<+\infty$ and compute
\begin{align*}
 \frac 1 p \frac d {dt} \| u(t) \|_p^p &\le - \frac {2(p-1)} p \int u |u|^p dx + \frac {p-1} p \int (e^{-|x|}*u) |u|^p dx
\end{align*}
Since $u\ge 0$ we can drop the first term and continue to estimate
\begin{align*}
 \frac 1 p \frac d {dt} \| u(t) \|_p^p \le C \|u_0\|_2 e^{Ct} \| u(t) \|_p^p,
\end{align*}
where $C$ is a generic constant independent of $p$. Gronwall then
implies that
\begin{align} \label{eq1220A}
 \| u(t) \|_p \le \|u_0\|_p \exp \left( Ct \|u_0\|_2 e^{Ct} \right).
\end{align}
Letting $p\to \infty$, we obtain
\begin{align} \label{eq1220B}
 \| u(t) \|_\infty \le \|u_0\|_\infty \exp \left( Ct \|u_0\|_2 e^{Ct} \right).
\end{align}
This is the $L_x^\infty$-norm estimate we needed.

\texttt{Step 2}: Control of $\| \partial_x u(t) \|_\infty$. We shall apply Lemma \ref{lemM} with
\begin{align*}
 & b(t,x,z,p) =z^2 +\epsilon, \\
& a(t,x,z,p) = 2zp^2 + 2 z^2 -c(t,x) z -d(t,x)p,
\end{align*}
where 
\begin{align*}
 c(t,x) = (\partial_x K*u_x)(t,x),\\
d(t,x) = (\partial_x K*u)(t,x).
\end{align*}
By step 1 we have
\begin{align*}
 |c(t,x)|+ |d(t,x)| &\le C \|u(t)\|_{\infty}\\
 & \le C \|u_0\|_{\infty} exp(CT\|u_0\|_2 e^{CT}).
\end{align*}
Also
\begin{align*}
 |u(t,x) -u(t,y)
 | &\le 2\|u(t)\|_\infty \le 2 \|u_0\|_\infty \exp\left( CT \|u_0\|_2 e^{CT} \right). \\
& =: M
\end{align*}
and
\begin{align*}
 |u_0(x) -u_0(y) | &\le \| \partial_x u_0\|_\infty |x-y| \\
&=: L_0|x-y|.
\end{align*}
Collecting all these estimates, we see that
\begin{align*}
 \sup_{0\le t\le T} \| \partial_x u(t) \|_\infty \le C(\|u_0\|_2, \|\partial_x u_0\|_\infty, \epsilon, T)
\end{align*}
This concludes the gradient estimate.

\texttt{Step 3}: Control of the higher derivatives. This part of the
estimates is exactly the same as the corresponding estimates in the
proof of Lemma \ref{lem301}. Note in particular that we do not need
$T$ to be small once we obtain the a priori control of $\|u\|_p$ and
$\|\partial_x u\|_\infty$ norms.
\end{proof}

Now we are ready to complete

\begin{proof}[Proof of Proposition \ref{prop2}]
This follows directly from the local existence Lemma
\ref{lem_local_exist} and the a priori estimate Lemma \ref{lem301B}.
If $u_0 \in L_x^p(\R)$ for some $1\le p<2$, then by repeating the
proof of Corollary \ref{cor1} we obtain $u\in C([0,T], L_x^p)$ for
any $T>0$. The proof is finished.
\end{proof}

\section{Proof of Theorem \ref{thm1}, \ref{thm2} and \ref{thm3}}

In this section we prove our main theorems. Our
solution is going to be the limit of the sequence of regularized
solutions $u^{\eps}$ which we constructed in the previous section.
To obtain uniform control of the Sobolev norms of the regularized solutions, we
have the following

\begin{lem}[A priori estimate] \label{lem_apriori}
Assume $u_0 \in H_x^{2k+1}(\R)$, $k\ge 2$. Let $T_0=T_0(\| u_0\|_2+
\| \partial_x u_0 \|_\infty)$ and $u^{\epsilon} \in B^k_{T_0}$ be
the corresponding unique solution to \eqref{eq_WSR_1} (see Lemma
\ref{lem301}). Then the set of functions
$(u^{\epsilon})_{0<\epsilon<1}$ satisfies the following uniform
estimates:
\begin{align*}
 & \sup_{0\le t\le T_0}\| u^\epsilon \|_{H^{2k+1}}
+\sqrt{\epsilon} \| u^{\epsilon} \|_{L_t^2 H^{2k+2}([0,T_0]\times \R)}
 \le C( \| u_0\|_{H^{2k+1}}, k), \\
& \| u^\epsilon \partial_x^{2k+2} u^\epsilon \|_{L_t^2 L_x^2([0,T_0]\times \R)} + \| \partial_t u^\epsilon\|_{L_t^2 H^{2k}([0,T_0]\times \R)}
\le C( \| u_0\|_{H^{2k+1}}, k),\\
& \| \partial_x (\partial_x K*u^{\epsilon} u^\epsilon)
 \|_{L_t^2 H^{2k}([0,T_0]\times \R)} \le C( \| u_0\|_{H^{2k+1}}, k),
\end{align*}
where $C=C(\| u_0\|_{H^{2k+1}}, k)$ is a positive constant.
\end{lem}
\begin{proof}
The estimates of $\|u^{\epsilon} \|_{L_t^\infty H^{2k+1}}$,
$\sqrt{\epsilon} \| u^{\epsilon} \|_{L_t^2H^{2k+2}}$ and $ \|
u^\epsilon \partial_x^{2k+2} u^\epsilon \|_{L_t^2 L_x^2}$ can be
recovered from step 4 of the proof of Lemma \ref{lem301}. To bound
$\| \partial_t u^\epsilon \|_{L_t^2 H^{2k}}$, we use
\eqref{eq_WSR_1} to obtain (here we again drop the superscript
$\eps$ for simplicity):
\begin{align*}
 \| \partial_t u \|_{H^{2k}} & \le \| \partial_t u
  \|_2 + \| \partial_x^{2k} \partial_t u \|_2 \\
& \le \| (u^2+\epsilon) u_{xx} + 2 u u_x^2 - \partial_x ( K*\partial_x u u) \|_2 + \| \partial_x^{2k} (u^2 u_{xx}) \|_2 \\
&\quad + \epsilon \| \partial_x^{2k+2} u \|_2 + 2 \| \partial_x^{2k} (u u_x^2) \|_2 + \| \partial_x^{2k+1}( K*\partial_x u u) \|_2 \\
& \le \| u\|_{H^2}^3+ {\| u\|_{H^2}^2} + \| u^2 \partial_x^{2k+2} u\|_2 + \| u\|_{H^{2k+1}}^3 \\
& \quad + \epsilon \| \partial_x^{2k+2} u\|_2 + \| u\|_{H^{2k+1}}^3 + \| u\|_{H^{2k+1}}^2.
\end{align*}
Integrating over $[0,T_0]$, we have
\begin{align*}
 \int_0^{T_0} \| \partial_t u(t)\|_{H^{2k}}^2 dt & \lsm T_0 \sup_{0\le t\le T_0} \| u\|_{H^{2k+1}}^2 (1+ \| u\|_{H^{2k+1}} ) \\
& \quad + \int_0^{T_0} \| u^2 \partial_x^{2k+2} u(t) \|_2^2 dt + \epsilon \int_0^T \| \partial_x^{2k+2} u \|_2^2 dt \\
& \le C(T_0,k, \|u_0\|_{H^{2k+1}}).
\end{align*}
Since $T_0=T_0(\|u_0\|_2+\|\partial_x u_0\|_\infty)$, we obtain the desired bounds. The last estimate is a simple
application of H\"older and Young's inequality using the given estimates.
\end{proof}

Now we define the set
\begin{align*}
 D^k_{T_0} := &\Bigl\{ v \in C([0,T_0]; H_x^{2k}) \cap L^\infty(0,T_0; H_x^{2k+1})):\, \Bigr.\\
&\Bigl. \partial_t v \in L^2(0,T_0; H^{2k})),\, v^2\partial_{xx}v
\in L^2(0,T_0; H^{2k}) \Bigr\}.
\end{align*}
We shall prove the existence and the uniqueness of the solution of
\eqref{eq1} in this set. This is
\begin{thm} \label{thm_local}
 Assume $u_0 \in H^{2k+1}(\R)$ with $k\ge 2$. Then there exists $T_0=T_0(\|u_0\|_2 + \|\partial_x u_0 \|_\infty)$ such that
\eqref{eq1} has a unique solution $u \in D^k_{T_0}$.
\end{thm}
\begin{proof}

\texttt{Step 1}. Contraction in $C([0,T],L_x^2)$. For any two
$\epsilon_1, \epsilon_2>0$, let $u=u^{\epsilon_1}$,
$v=u^{\epsilon_2}$ solve \eqref{eq_WSR_1} with the same initial data $u_0$.
Denote $w=u-v$, then for $w$ we have the equation
\begin{align*}
 \partial_t w &= \frac 13 (u^3-v^3)_{xx} + \epsilon_1 u_{xx} -\epsilon_2 v_{xx} - \partial_x (K*\partial_x w u) \\
&\qquad - \partial_x (K*\partial_x v w).
\end{align*}
Then
\begin{align*}
 \frac 12 \frac d {dt} \| w(t) \|_2^2 & = \frac 13 \int (u^3-v^3)_{xx} w dx + \int (\epsilon_1 u_{xx} -\epsilon_2 v_{xx}) wdx \\
&\quad -\int \partial_x(K*\partial_x w u) w dx  - \int \partial_x (K*\partial_x v w) wdx \\
& =: I +II+III+IV.
\end{align*}
Denote $F(u,v)=u^2+v^2+uv$. Clearly we always have $F(u,v)\ge 0$. We compute
\begin{align*}
 I & = \frac 13 \int \partial_{xx}(w F(u,v)) w dx \\
& = \frac 13 \int \partial_{xx} w w F(u,v) dx \\
& = - \frac 13 \int (\partial_x w)^2 F(u,v) dx + \frac 16 \int w^2  \partial_{xx} F(u,v) dx \\
& \le \frac 16 \int w^2 \partial_{xx} F(u,v) dx \\
& \lsm \| w\|_2^2 ( \| u\|_{H^3}^2 + \|v\|_{H^3}^2).
\end{align*}
For II we have the estimate
\begin{align*}
 II \le (\epsilon_1 \|u\|_{H^2} + \epsilon_2 \|v \|_{H^2}) \| w\|_{2}.
\end{align*}
The term III can be bounded as:
\begin{align*}
 III & = \int K*\partial_x w u w_x dx \\
&= -\int \partial_{xx}K*w uw dx -\int K*\partial_x w u_x w dx \\
& \le (\|u\|_\infty + \|\partial_x u\|_\infty) \| w\|_2^2.
\end{align*}
Similarly for IV we get
\begin{align*}
 IV & \lsm \| K*\partial_{xx} v\|_\infty \| w\|_2^2 \\
& \lsm \|v \|_{H^1} \|w\|_2^2.
\end{align*}
Finally we have by Lemma \ref{lem_apriori}
\begin{align*}
 \frac d {dt} \| w(t)\|_2^2 &\lsm \| w(t) \|_2^2 ( \| u\|_{H^3}^2 + \| v\|_{H^3}^2 + \|u\|_{H^3} + \| v\|_{H^3} ) \\
&\quad + (\epsilon_1 \| u\|_{H^2} + \epsilon_2 \| v\|_{H^2} ) \|w(t)\|_{2} \\
& \lsm C(\|u_0\|_{H^{2k+1}},k) \|w(t)\|_2^2 +
C(\|u_0\|_{H^{2k+1}},k) (\epsilon_1+\epsilon_2) \|w(t)\|_{2}
\end{align*}
Gronwall then gives
\begin{align*}
 \| u^{\epsilon_1}(t)-u^{\epsilon_2}(t) \|_2 &= \|w(t)\|_2 \\
& \le C(\|u_0\|_{H^{2k+1}},k) (\epsilon_1+\epsilon_2) \exp \left( C(\|u_0\|_{H^{2k+1}},k) t \right).
\end{align*}
This implies that $(u^\epsilon)_{0<\epsilon<1}$ has a limit as
$\epsilon\to 0$ in $C([0,T_0];\, L_x^2)$. Denote this limit as $u
\in C([0,T_0];\, L_x^2)$.

\texttt{Step 2}. We show that $u$ is our desired solution in
$D^k_{T_0}$. By Lemma \ref{lem_apriori} and the interpolation
inequality 
\begin{align*}
 \| f\|_{H^{2k}} \lsm \|f\|_2^{\frac 1 {2k+1}} \| f\|_{H^{2k+1}}^{\frac {2k}{2k+1}},
\end{align*}
we obtain
\begin{align*}
 u^{\epsilon} \to u \quad \text{in}\; C([0,T_0];\, H_x^{2k}), \quad \text{as $\epsilon\to 0$}.
\end{align*}
It then follows easily that
\begin{align*}
 & (u^\epsilon)^2 \partial_{xx} u^{\epsilon} \to u^2 \partial_{xx} u, \quad
 \text{in $C(0,T_0; L_x^2)$},\\
&  u^\epsilon (\partial_x u^{\epsilon})^2 \to u (\partial_x u)^2, \quad \text{in $C(0,T_0;
L_x^2)$},\\
& \partial_x(K*\partial_x u^{\epsilon} u^\epsilon ) \to \partial_x
(K*\partial_x u u), \quad \text{in $C(0,T_0; L_x^2)$}.
\end{align*}
By Lemma \ref{lem_apriori}, the functions $(u^\epsilon)^2 \partial_{xx}
u^\epsilon$, $u^\epsilon (\partial_x u^{\epsilon})^2$, $\partial_x
(K*\partial_x u^\epsilon u^\epsilon)$ are uniformly bounded in
$L^2([0,T_0], H_x^{2k})$. We conclude that $u^2 \partial_{xx} u$, $u
(\partial_x u)^2$, $\partial_x (K*\partial_x u u) \in
L^2([0,T_0],H_x^{2k})$. Now since $u^\epsilon$ is also uniformly
bounded in $L^\infty([0,T_0],H_x^{2k+1})$, we obtain $u \in
L^\infty([0,T_0],H_x^{2k+1})$. Similarly by Lemma \ref{lem_apriori}, the
set of functions $ (\partial_t u^{\epsilon})_{0<\epsilon<1}$ is uniformly bounded in $L^2([0,T_0],
H_x^{2k})$. By extracting a subsequence if necessary, we obtain for
some $\epsilon_n \to 0$,
\begin{align*}
 \partial_t u^{\epsilon_n} \rightharpoonup \partial_tu, \quad \text{in $L^2([0,T_0],H_x^{2k})$}.
\end{align*}
It follows that $\partial_t u \in  L^2([0,T_0],H_x^{2k}) $. Finally
by using Lemma \ref{lem_apriori} again we have
\begin{align*}
 \epsilon_n \partial_{xx} u^{\epsilon_n} \to 0, \quad \text{in $L^2(0,T_0;H_x^{2k})$}.
\end{align*}
This shows that $u$ is our desired solution in $D^k_{T_0}$. The theorem is proved.
\end{proof}


\begin{cor}[Continuation of solutions] \label{cor_cont}
Let $u_0 \in H_x^{2k+1}(\R)$ with $k\ge 2$. Then there exists a
unique solution $u$ of \eqref{eq1} with maximal-lifespan $T^*$ such
that only one of the following possibilities occur:
\begin{enumerate}
 \item $T^*=\infty$ and $u \in D^k_T$ for any $T>0$.
 \item $T^*<\infty$, $u\in D^k_T$ for any $T<T^*$, and
\begin{align*}
 \lim_{t\to T^*} \| \partial_x u(t) \|_\infty =\infty.
\end{align*}
\end{enumerate}
For any $2\le p<\infty$, there exists a generic constant $C$ such that
\begin{align*}
 \| u(t) \|_p \le \|u_0 \|_p e^{Cpt}, \quad \forall\, t\ge 0.
\end{align*}
If in addition $u_0\ge 0$, then $u(t)\ge 0$ and we also have the $p$-independent estimate for
all $2\le p\le +\infty$:
\begin{align} \label{eq1219}
 \| u(t) \|_p \le \|u_0\|_p \exp \left ( Ct \|u_0\|_2 e^{Ct} \right), \quad\forall\, t\ge 0.
\end{align}
In particular if $p=+\infty$, then
\begin{align*}
 \| u(t) \|_{\infty} \le \|u_0\|_\infty \exp\left( Ct \|u_0\|_2 e^{Ct} \right), \quad\forall\, t\ge 0.
\end{align*}
\end{cor}

\begin{proof}
This follows directly from Theorem \ref{thm_local} and Lemma
\ref{lem301} (see \eqref{eq948} for the growth estimate of 
$L^p_x$-norm for any $T>0$). Since $u^\epsilon \to u$ uniformly in $C([0,T];
H_x^{2k})$, $u(t)$ satisfies the same estimate \eqref{eq948}.
Similarly if $u_0\ge 0$, then the estimate \eqref{eq1219} follows
from the corresponding estimate for $u^\epsilon$ (see
\eqref{eq1220A} and \eqref{eq1220B} in the proof of Lemma
\ref{lem301B}). Now that the $L^2$ norm of the constructed local
solution remains finite for any $T>0$, it can be continued as long
as $\| \partial_x u(t) \|_\infty$ does not blow up.
\end{proof}

We are now ready to complete the
\begin{proof}[Proof of Theorem \ref{thm1}]
 By Theorem \ref{thm_local}, $u_0 \in \bigcap_{m=0}^\infty H_x^m (\R)$
 implies that $u \in \bigcap_{k=2}^\infty D^k_{T_0}$.
Therefore $u\in C([0,T_0], H_x^m)$ for any $m\ge 0$.  It also
follows that $\partial_t u \in \bigcap_{m=1}^\infty
L^2(0,T_0;H^m_x)$. By using \eqref{eq1}, a standard bootstrap
argument implies that $u\in C^\infty([0,T_0]\times \R)$. If $u_0\ge
0$, then the positivity of $u$ follows from Lemma \ref{lem301B} and
the uniform convergence of $u^\epsilon$ to $u$. The theorem is
proved.
\end{proof}

\begin{proof}[Proof of Theorem \ref{thm2}]
This follows immediately from Theorem \ref{thm1} and Corollary \ref{cor_cont}.
\end{proof}

\begin{proof}[Proof of Theorem \ref{thm3}]
This is only slightly different from the proof of Corollary \ref{cor1}.
Assume $u_0 \in \bigcap_{m=0}^\infty H_x^m(\R) \bigcap L_x^p(\R)$
for some $1\le p<2$.
Then using \eqref{eq1}, we can bound
\begin{align*}
 \| u(t) \|_p &\lsm \| u_0\|_p + \int_0^t \left \|  \frac 13 \partial_{xx}(u^3) - \partial_x( (\partial_x K*u) u) \right\|_p ds \\
& \lsm \| u_0\|_p + \int_0^t \| u(s) \|_{W^{2,3p}}^3+ \| u(s) \|_{W^{1,2p}}^2 ds \\
& \lsm \| u_0\|_p + \int_0^t \| u(s)\|_{H^3}^3 + \|u(s) \|_{H^2}^2  ds.
 \end{align*}
 This shows that $u(t) \in L_x^p$ for any $t$. The continuity (including right continuity at $t=0$) follows from similar estimates.
We omit the standard details.
Finally in the case $u_0\ge 0$, $u_0 \in L_x^1(\R)$, the $L^1_x$ preservation follows from direct integration.
\end{proof}

\section{Proof of Theorem \ref{thm4}}
 We argue by contradiction. Let $u_0 \in
C_c^\infty(\R)$, $u_0\ge 0$ and assume that the corresponding
solution $u$ is global. By Theorem \ref{thm1}, \ref{thm2},
\ref{thm3}, $u \in C([0,T), H_x^m)$ for any $m\ge 0$ and the
$L_x^1$-norm of $u$ is preserved. Our intuition of proving the
blowup is based on the observation that as time goes on, the boundary
of the solution(which in $1$D case consists of two points) will move face to
face at a certain speed which has a lower bound independent of time.
This clearly will lead to the collapse of the solution. To realize this intuition, we will carry out a detailed
analysis on the characterstic lines of the solution which satisfy
\begin{align} \label{eqXta}
 \begin{cases}
  \dfrac d {dt} X(t, \alpha) =  (K*\partial_x u) (X(t,\alpha),t), \\
X(0,\alpha) = \alpha \in \R.
 \end{cases}
\end{align}
By standard ODE theory and our assumption that $u$ is a smooth global solution, $X(t)$ is
well defined and smooth for all time. Moreover, we have the
following lemma.
\begin{lem}[Properties of $X(t,\alpha)$] \label{lem_Xta}
The characterstic lines $X(t,\alpha)$ and the solution $u(t,x)$
satisfy the following properties:
\begin{enumerate}
 \item For any $t\ge 0$, $X(t,\cdot):\R\to\R$ is a $C^\infty$ diffeomorphism.
 \item $X(t,\cdot)$ is an order-preserving map in the sense that if $\alpha_1<\alpha_2$, then $X(t,\alpha_1)<X(t,\alpha_2)$ for any $t\ge 0$.
 \item $X(t,\cdot)$ maps intervals to intervals. More precisely, for any $\alpha_1<\alpha_2$, denote
          $$X(t,\cdot)([\alpha_1,\alpha_2])=\left\{ X(t,\alpha): \alpha_1 \le \alpha \le \alpha_2 \right\}.$$
       Then $X(t,\cdot)([\alpha_1,\alpha_2]) = [X(t,\alpha_1], X(t,\alpha_2)]$.
 \item If $u_0(\alpha)=0$, then $u(X(t,\alpha),t)=0$ for any $t\ge 0$.
 \item If $supp(u_0(\cdot)) \subset [-L,L]$ for some $L>0$, then
\begin{align*}
supp(u(t,\cdot)) \subset [X(t,-L),X(t,L)] \subset [-L,L]
\end{align*}
for all $t\ge 0$.
\end{enumerate}
\end{lem}
\begin{proof}
Property (1) is rather standard. Note in particular that
\begin{align*}
 \frac {\partial X(t,\alpha)} {\partial \alpha} = \exp \left(  \int_0^t (\partial_x K*\partial_x u)(X(s,\alpha),s) ds \right)
\end{align*}
is smooth and globally invertible due to the fact that $u$ is
smooth. Property (2) is also trivial. Assume not true, then by
the intermediate value theorem, for some $\tau >0$ we must
have $X(\tau,\alpha_1)-X(\tau,\alpha_2)=0$ which is a contradiction
to property (1). Property (3) follows immediately from property (1)
and (2). For property (4), one simply uses \eqref{eq1} and
\eqref{eqXta}. Writing $u(t,\alpha) = u(X(t,\alpha),t)$, we have
\begin{align*}
\begin{cases}
\frac d {dt }u(t,\alpha) = \left(-\partial_x K*\partial_x u + u \partial_{xx} u+ 2 (\partial_x u)^2 \right) u(t,\alpha), \\
u(0,\alpha) = u_0(\alpha).
\end{cases}
\end{align*}
On the RHS of the above equation, one can regard the coefficient of
$u(t,\alpha)$ as given functions of $(t,\alpha)$. By direct integration, property (4) now follows immediately. For property (5) we
take $y<x(t,-L)$, then by property (1) and (2) we have
$y=X(t,\alpha)$ for some $\alpha<-L$. By property (4) since
$u_0(\alpha)=0$, we must have $u(t,y)=u(t,X(t,\alpha))=0$. Similarly
once can show that if $y>X(t,L)$ then $u(t,y)=0$. This shows that
$supp(u(t,\cdot)) \subset [X(t,-L),X(t,L)]$. Now using this fact and
\eqref{eqXta}, we compute
\begin{align} \label{eq4459}
 \frac d {dt} X(t,-L)  & = (\partial_x K*u)(t, X(t,-L))  \notag\\
& = - \int_{\R} sign(X(t,-L)-y) e^{-|X(t,-L)-y|} u(t,y) dy \notag \\
& = - \int_{y\ge X(t,-L)} sign(X(t,-L)-y) e^{-|X(t,-L)-y|} u(t,y) dy \notag\\
& \ge 0.
\end{align}
By a similar computation one can show that $\frac d {dt} X(t,L) \le 0$. These facts together with property 2 and 3 above easily yield
that $[X(t,-L), X(t,L)] \subset [-L,L]$.
\end{proof}
Using lemma \ref{lem_Xta}, we are able to finish the
\begin{proof}[Proof of Theorem \ref{thm4}]
We argue by contradiction. Let $u_0 \in C_c^\infty(\R)$, $u_0 \ge 0$
with $supp(u_0) \subset [-L,L]$ for some $L>0$. Assume that the
corresponding solution $u$ is global. Define the characteristic
$X(t,\alpha)$ according to \eqref{eqXta}. By Lemma \ref{lem_Xta} we
have
\begin{align} \label{eq5503}
supp(u(t,\cdot)) \subset [X(t,-L),X(t,L)] \subset [-L,L],\quad \forall\, t\ge 0.
\end{align}
By this fact, we then compute using \eqref{eq4459},
\begin{align*}
  \frac d {dt} X(t,-L)  &= - \int_{y\ge X(t,-L)} sign(X(t,-L)-y) e^{-|X(t,-L)-y|} u(t,y) dy \\
 & = -\int_{ X(t,-L) \le y \le X(t,L)} sign(X(t,-L)-y) e^{-|X(t,-L)-y|} u(t,y) dy\\
& = \int_{ X(t,-L) \le y \le X(t,L)} e^{-|X(t,-L)-y|} u(t,y) dy\\
& \ge e^{-2L} \|u_0\|_{L^1},
\end{align*}
where in the last step we used the fact that the $L^1$ norm of $u$
is preserved. This shows that $X(t,-L)$ grows linearly with time
which is contradiction to \eqref{eq5503}. Thus we have shown that
the solution $u$ with $u_0$ as initial data cannot exist globally.
Finally \eqref{eq_thm41} and \eqref{eq_thm42} are easy consequences
of Theorem \ref{thm1}, \ref{thm2} and \ref{thm3}. The theorem is
proved.
\end{proof}

\section*{Acknowledgements}
D. Li and  X. Zhang are supported under NSF grant DMS-0635607. X.
Zhang is also supported under NSF grant No. 10601060, Project 973
and the Knowledge Innovation Program of CAS.

\bibliographystyle{amsplain}

\end{document}